\documentclass[12pt]{article}
\usepackage{times}
\usepackage{fullpage}
\usepackage{url}
\usepackage{amsmath,amssymb,amsthm,mathtools}
\usepackage[ruled,vlined,linesnumbered]{algorithm2e}
\usepackage{tikz,pgfplots,pxpgfmark}
\usetikzlibrary{calc,shapes.callouts,decorations.text,shadows}
\pgfplotsset{compat=1.14}

\usepackage{comment}

\newcommand{\ot}{\leftarrow}
\newcommand{\argmax}{\mathop{\rm arg\,max}}

\newcommand{\Ch}{\mathop{\mathrm{Ch}}\nolimits}
\newcommand{\upperW}{\overline{w}}
\newcommand{\lowerW}{\underline{w}}

\newenvironment{modified}{}{}

\newtheorem{theorem}{Theorem}
\newtheorem{lemma}{Lemma}

\newtheorem{definition}{Definition}
\newtheorem{example}{Example}
\providecommand{\keywords}[1]
{
  \small	
  \textbf{\textit{Keywords---}} #1
}
\providecommand{\jel}[1]
{
  \small	
  \textbf{\textit{JEL Classification---}} #1
}

\begin{document}

\title{Near-Feasible Stable Matchings with Budget Constraints
}

\author{{\bf Yasushi Kawase}\\
  {\normalsize Tokyo Institute of Technology and} \\
  {\normalsize RIKEN AIP Center, Tokyo, Japan.}\\
  {\normalsize \texttt{kawase.y.ab@m.titech.ac.jp}}
  \and 
  {\bf Atsushi Iwasaki}\\
  {\normalsize University of Electro-Communications and} \\
  {\normalsize RIKEN AIP Center, Tokyo, Japan.}\\
  {\normalsize \texttt{iwasaki@is.uec.ac.jp}}
}

\date{}
\maketitle

\begin{abstract}
  We consider the matching with contracts framework of Hatfield and Milgrom~\cite{Hatfield:AER:2005} when
  one side (a firm or hospital) can make monetary transfers (offer wages) to the other (a worker or doctor). 
  In a standard model, monetary transfers are not restricted. However, we assume that each hospital has a {\em fixed budget}; that is, 
  the total amount of wages allocated by each hospital to the doctors is constrained. 
  With this constraint, stable matchings may fail to exist and checking for the existence is hard. 
  To deal with the nonexistence,
  we focus on \textit{near-feasible} matchings that can exceed each hospital budget by a certain amount, and 
  we introduce a new concept of \textit{compatibility}. We show that the compatibility condition is 
  a sufficient condition for the existence of a near-feasible stable matching in the matching with contracts framework. 
  Under a slight restriction on hospitals' preferences, we provide mechanisms that efficiently return
  a near-feasible stable matching with respect to the actual amount of wages allocated by each hospital. 
  By sacrificing strategy-proofness, the best possible bound of budget excess is achieved. 

\end{abstract}
\\
\jel{C78, D47}

\section{Introduction}
This paper studies a two-sided, one-to-many matching model
where there are budget constraints on one side (a firm or hospital), that is, 
the total amount of wages that it can pay to the other side (a worker or doctor) is limited.
The theory of two-sided matching has been extensively developed, as is illustrated by
the comprehensive surveys of Roth and Sotomayor~\cite{Roth:CUP:1990} or Manlove~\cite{manlove:2013}. 
Rather than fixed budgets, maximum quotas limiting 
the total number of doctors that each hospital can hire are typically used. 
Some real-world examples are subject to the matching problem with budget constraints: 
a college can offer stipends to recruit better students when the budget for admission is limited; 
a firm can offer wages to workers under the condition that employment costs depend on earnings in the previous accounting period; 
a public hospital can offer salaries to doctors in the case where the total amount relies on funds from the government; and so on. 
To establish our model and concepts, we use doctor-hospital matching as a running example. 

To date, most papers on matching with monetary transfers assume that budgets are unrestricted, e.g.,~\cite{kelso:ecma:1982}. 
When they are restricted, stable matchings may fail to exist~\cite{mongell:el:1986,abizada:te:2016}.
In fact, Abizada~\cite{abizada:te:2016} considers a subtly different model from ours 
and shows that (coalitional) stable matchings, 
where groups of doctors and hospitals have no profitable deviations, may not exist.%
 
There are several possibilities to circumvent the nonexistence problem. 
Abizada~\cite{abizada:te:2016} modifies the notion of stability and
proposes a variant of the deferred acceptance mechanism that produces
a pairwise, instead of coalitional, stable matching, and is strategy-proof for doctors. 
Dean, Goemans, and Immorlica~\cite{dean:ifip:2006}
examine a problem similar to ours, 
restricting each hospital to having only a \textit{lexicographic} utility. 
Nguyen and Vohra~\cite{nguyen:ec:2015} examine matchings with couples 
where joint preference lists over pairs of hospitals are submitted, 
e.g., \cite{
  kojima:qje:2013,perrault:aamas:2016}.
They then develop an algorithm that outputs a \textit{near-feasible} stable matching, 
in which the number of doctors assigned to each hospital
differs (up or down) from the actual maximum quota by at most three.

We focus on near-feasible matchings on a different model from Nguyen and Vohra~\cite{nguyen:ec:2015} 
that exceeds the budget of each hospital by a certain amount. 
This idea can be interpreted as one in which, for each instance of a matching problem, 
our mechanisms find a ``nearby'' instance with a stable matching. 
However, it must be emphasized that Nguyen and Vohra 
discuss no strategic issue, i.e., misreporting a doctor's preference may be profitable. 
The literature on matching has found strategy-proofness for doctors, that is, 
no doctor has an incentive to misreport his or her preference, 
to be a key condition in a wide variety of settings~\cite{Abdulkadiroglu:AER:2003}.

The stability and strategy-proofness on matching mechanisms with/without monetary transfers
is characterized by the \emph{generalized Deferred Acceptance (DA) mechanism} on a ``matching with contracts'' model~\cite{Hatfield:AER:2005}. 
When budgets are unrestricted, if a mechanism---the choice function of every hospital---satisfies 
three conditions, specifically, \textit{substitutability}, \textit{irrelevance of rejected contracts}, and \textit{law of aggregate demand},
then it always finds a ``stable'' allocation and is strategy-proof for doctors. 
What happens if they are restricted? 
No choice function that simply chooses the most preferred contracts does not satisfy the substitutability condition. This coincides with the fact that no stable matching exists in the presence of budget constraints. With such a choice function, a hospital can form a blocking coalition to increase its utility. Besides, without violating budget constraints, no choice function cannot be designed so that it satisfies the substitutability condition and outputs the best set among the doctors offered to each hospital via the generalized DA process. We treat such constraints as ``soft'' and characterize the stability and strategy-proofness on near-feasible matching mechanisms. To this end, we introduce a new additional condition, which we call \textit{compatibility}, on the matching with contract model.

The compatibility condition ensures that a choice function chooses the best set of contracts among the offered contracts under a certain capacity that is greater than or equal to the predefined capacity. We show that the four conditions: substitutability; irrelevance of rejected contracts; law of aggregate demand; and compatibility are the sufficient conditions for choice functions so that, for any instance of matching problems with budget constraints, there is a nearby instance guaranteed to have a stable matching. Note that the compatibility condition does not specify how much a hospital may exceed its budget and just suggests a direction for choice functions for the desirable conditions. It is still challenging that we construct a choice function that satisfies the four conditions without violating the budget constraints too much.

Alternatively, from a practical point of view, we need to compute choice functions efficiently. 
However, computing each hospital's choice function under a certain budget is NP-hard because it is equivalent to solving a knapsack problem. 
In fact, 
McDermid and Manlove~\cite{mcdermid2010} 
consider matchings with inseparable couples. 
In this problem, there are single doctors and couples. 
Every single doctor and couple has a preference over hospitals
and each hospital has a maximum quota.
Their model is contained in ours as a special case. 
They prove that deciding whether a stable matching exists or not is NP-hard (Non-deterministic Polynomial-time hard). 
In contrast, the choice functions we develop can be implemented to run fast via consulting Dantzig's greedy algorithm~\cite{dantzig1957}. 


Our analysis is general particularly because each hospital's priority is represented by the weak relation over the subsets of contracts as many papers have done, e.g.,~\cite{kojima:qje:2013}. It would be potentially useful when you consider constraints with a similar structure to budgets, and practically useful in designing matching mechanisms to work. To see this point, we develop two mechanisms that can be implemented to run fast: 
one is strategy-proof for doctors and the other is not. 
Our first mechanism outputs a stable matching that exceeds 
each budget by a factor at most the maximum wage each hospital offers over the minimum one. 
Our second mechanism, by sacrificing strategy-proofness, improves the violation bound; 
that is, it produces a stable matching that exceeds 
each budget up to the maximum wage each hospital offers. This is the best possible bound.
The key procedure of our mechanisms, i.e., choice functions of hospitals, is sorting contracts in the decreasing order of the hospital's utility per wage. We thus allow each hospital to have an \textit{additive} utility over contracts. Note that we could build a mechanism when each hospital has a weak relation, but the bound may reach an unreasonable amount. 
Furthermore, we find two special cases such that a slight modification of 
the second mechanism recovers strategy-proofness, keeping the best possible bound.


Our model assumes that the amount of predefined budgets is flexible up to a certain amount. 
There is certainly some realistic situation where this assumption is justified. 
Indeed, in firm-worker matchings,
if a firm receives an application from a worker who is appropriate for the business, 
the CEO would agree to increase the employment cost.
In doctor-hospital matchings, 
hospitals can create an association that pools funds in advance and subsidizes 
the expense of salaries according to matching results. 
Even when budgets must not be exceeded, 
our mechanisms can work by decreasing budgets in a certain amount in advance. 
Our second mechanism produces a near-feasible matching that does not exceed the predefined budgets. 
 
Alternatively, recall that our model contains a matching problem with inseparable couples~\cite{mcdermid2010} in which two doctors of each couple must be assigned to the same hospital. 
We can consider this problem as ours with just two possible wages, e.g., $1$ or $2$, and budgets corresponding to maximum quotas. 
Our second mechanism outputs a stable matching that exceeds each maximum quota up to $1$. 
Notice that the matching mechanism of Nguyen and Vohra~\cite{nguyen:ec:2015} is not easily extended to our setting 
since it is designed for \textit{separable} couples. Even if it could be extended, their mechanism is computationally hard
because it relies on Scarf's Lemma, whose computational version is PPAD-complete~\cite{kintali::2008}
(Polynomial Parity Arguments on Directed graphs). 
Thus, the mechanism is unlikely to be implemented to run fast, particularly when many possible contracts are offered. 

Let us briefly explore some remaining related literature. 
Matching with constraints is prominent across computer science and economics~\cite{kamada:aer:2017}. 
In many application domains, various distributional constraints are 
often imposed on an outcome, e.g., regional maximum quotas are imposed 
on hospitals in urban areas to allocate more doctors to
rural areas~\cite{kamada:aer:2015} or minimum quotas are imposed when 
school districts require that 
at least a certain number of students are allocated to each school
to enable these school to operate properly%
~\cite{biro:tcs:2010,fragiadakis::2012,Goto:aamas:2014,aziz:aamas:2019}.
%
Yet another type of distributional constraint involves  
\textit{diversity constraints} in school choice programs~\cite{Abdulkadiroglu:AER:2003,kojima2012school,hafalir2013effective,ehlers::2012}. 
They are implemented to give students/parents 
an opportunity to choose which public school to attend. 
However, a school is required to balance its composition of students,
typically in terms of socioeconomic status.
Controlled school choice programs must provide choices
while maintaining distributional constraints. 
Moreover, refugee settlement~\cite{delacretaz::2016} is an important application where a government assigns refugee families to some facilities it prepares. It is considered as a generalization of school choice because it needs to consider and satisfy several properties of each family. Those properties, such as member composition and status of children, are represented as 
(multi-dimensional) knapsack constraints. 
Although the complexity analysis is provided~\cite{aziz:aamas:2018}, it does not focus on near-feasible matchings. 
%
Finally, note that there is a certain amount of recent studies 
on two-sided matchings in the artificial intelligence and multi-agent systems community, 
although this literature has been established in the field across algorithms and economics. 
Drummond and Boutilier~\cite{drummond:ijcai:2013,drummond:aaai:2014} 
examine preference
elicitation procedures for two-sided matching.
In the context of \textit{mechanism design},
Hosseini, Larson, and Cohen~\cite{hosseini:aaai:2015} 
consider a mechanism for a situation in which 
agent preferences dynamically change.
Kurata et al.~\cite{kurata:jair:2017} 
explore strategy-proof mechanisms
for affirmative action in school choice programs (diversity constraints), while
Goto et al.~\cite{Goto:aij:2016} 
handle regional constraints, e.g., imposing 
regional minimum/maximum quotas on hospitals in urban areas 
so that more doctors are allocated to rural areas. 

The rest of this paper is organized as follows: 
In the next section, 
we set forth the notations and basic concepts for our two-sided matching model.
Section~\ref{sec:impossibility} provides impossibility and intractability results
and Section~\ref{sec:com} introduces a new condition
for when each hospital faces a budget constraint. 
Section~\ref{sec:mechs} develops two mechanisms that can be implemented to run fast: 
one is strategy-proof for doctors and the other is not. 
We further examine two special cases 
such that a slight modification of the second mechanism recovers
strategy-proofness, while maintaining the best possible bound. 
Finally, 
we present our conclusion in Section~\ref{sec:conclusion}. 

\section{Preliminaries}
\label{sec:preliminaries}
This section describes a model for two-sided matchings with budget constraints. 
A market is a tuple $(D,H,X,\succ_D,\succsim_H,B_H)$, where each component is defined as follows: 
There is a finite set of doctors $D=\{d_1,\ldots,d_n\}$ and a finite set of hospitals $H=\{h_1,\ldots,h_m\}$. 
Let $X\subseteq D\times H \times \mathbb{R}_{++}$ denote a finite set of contracts, where each contract $x\in X$ is of the form $x=(d,h,w)$. 
Here, $\mathbb{R}_{++}$ is the set of positive real numbers. 
A contract means that hospital $h\in H$ offers wage $w\in\mathbb{R}_{++}$ to doctor $d$. 
A hospital can choose a wage freely within $\mathbb{R}_{++}$ and 
can offer a doctor multiple contracts with different wages. 
Each contract is acceptable for each hospital~$h$. 
Furthermore, for any subset of contracts $X'\subseteq X$, 
let $X'_d$ denote $\{(d',h',w')\in X' \mid d'=d\}$ and $X'_h$ denote $\{(d',h',w')\in X' \mid h'=h\}$. 
We use the notations $x_D$, $x_H$, and $x_W$ to describe 
the doctor, the hospital, and the wage associated with a contract $x\in X$, respectively. 

Let $\succ_D=(\succ_d)_{d\in D}$ denote the doctors' preference profile,
where $\succ_d$ is the strict relation of $d\in D$ over $X_d\cup\{\emptyset\}$, that is, 
$x\succ_d x'$ means that $d$ strictly prefers $x$ to $x'$. $\emptyset$ indicates a \textit{null} contract. 

Let $\succsim_H=(\succsim_h)_{h\in H}$ denote the hospitals' preference profile,
where $\succsim_h$ is the weak relation of $h\in H$ over the subsets of contracts that includes at most one contract for each doctor, i.e., $\mathcal{X}_h\coloneqq \{X'\subseteq X_h\mid |X'_d|\le 1~(\forall d\in D)\}$.
Let $\succ_h$ and $\sim_h$ be strict and indifferent preference relation of $h$, respectively.
We call a market is \emph{additive} if the preference $\succsim_h$, for each $h\in H$, can be represented by an \emph{additive utility function} $f_h\colon \mathcal{X}_h\to\mathbb{R}_{+}$, that is,
(i) $X'\succsim X''$ if and only if $f_h(X')\ge f_h(X'')$ for any $X', X''\in \mathcal{X}_h$ and
(ii) $f_h(X')=\sum_{x\in X'}f_h(\{x\})$ holds for any $X'\in \mathcal{X}_h$.
In what follows, we simply denote $f_h(\{x\})$ by $f_h(x)$.



Each hospital $h$ has a fixed budget $B_h\in \mathbb{R}_{++}$ that it can distribute as wages to its admitted doctors. 
Let $B_H=(B_h)_{h\in H}$ be the budget profile. 
We assume that, for any contract $(d,h,w)$, $0< w \leq B_h$ holds. 
Given $X$, let 
\begin{align*}
 \lowerW_h=\min_{x\in X_h} x_W \quad\text{and}\quad \upperW_h=\max_{x\in X_h} x_W.
\end{align*}
Moreover, we use the notation $w_h(X')$ for any $X'\subseteq X$ to denote the total wage that $h$ offers in $X'$,
i.e., $\sum_{x\in X_h'}x_W$.

We call a subset of contracts $X'\subseteq X$ a \emph{matching} if $|X'_d| \le 1$ for all $d\in D$.
A matching $X'\subseteq X$ is \emph{$B'_H$-feasible} if
$w_h(X')\le B'_h$ for all $h \in H$. 
%
%
Given a matching $X'$, 
another matching $X''\in \mathcal{X}_h$ for hospital $h$ is a \textit{blocking} set (or coalition) if 
$X''\succ_{x_D} X'$ for all doctors $x_D$ of $x\in X''\setminus X'$, 
$X''\succ_h X_h'$, and 
$w_h(X'')\le B'_h$.
If such $X''$ exists, we say $X''$ blocks $X'$. 
Then we obtain a stability concept. 
\begin{definition}[$B'_H$-stability]\label{def:stability}
  We say a matching $X'$ is $B'_H$-\emph{stable} if
  it is $B'_H$-feasible and there exists no blocking coalition. 
\end{definition}


As we will see, 
when 
no hospital is allowed to violate the given constraints $B_H$, 
conventional stable matchings ($B'_H=B_H$) may not exist. 
Definition~\ref{def:stability}, 
for example, allows a central planner to add or redistribute the budgets 
and this planner finds a problem instance with $B'_H(\geq B_H)$, whose $B'_H$-stable matching is guaranteed to exist. 
If each contract has the same amount of wage $w$, we can regard a budget constraint $B_h$ 
for each hospital $h$ as its maximum quota of $\lfloor B_h/w\rfloor$.
Note also that a matching $X'\subseteq X$ is \textit{doctor-optimal} if $X'_d \succeq_{d} X''_d$ for all $d\in D$ and all $X''\subseteq X$. 

A mechanism is a function that takes a profile of 
doctor preferences as input and returns matching $X'$. 
We say a mechanism is stable if
it always produces a $B_H'$-stable matching for certain $B_H'$. 
We also say a mechanism is \emph{strategy-proof} for doctors if
no doctor ever has any incentive to misreport her preference,
regardless of what the other doctors report.

Next, we briefly describe a class of mechanisms called 
the generalized DA mechanism~\cite{Hatfield:AER:2005} and its conditions. 
This mechanism uses choice functions 
$\Ch_D\colon 2^X \rightarrow 2^X$ and 
$\Ch_H\colon 2^X \rightarrow 2^X$. 
For each doctor $d$, its choice function $\Ch_d(X')$ chooses $\{x\}$,
where $x=(d,h,w)\in X'_d$ such that
$x$ is the most preferred contract within $X'_d$
(we assume $\Ch_d(X')=\emptyset$ if $\emptyset\succ_d x$ for all $x\in X'_d$).
Then, the choice function of all doctors is given as:
$\Ch_D(X') \coloneqq \bigcup_{d \in D} \Ch_d(X_d')$.
Similarly, the choice function of all hospitals $\Ch_H(X')$ is
$\bigcup_{h \in H} \Ch_h(X_h')$, where $\Ch_h$ is a choice function of $h$.
There are alternative ways to define the choice function of each hospital $\Ch_h$. 
As we discuss later, the mechanisms considered in this paper are expressed 
by the generalized DA with different formulations of $\Ch_H$. 
Formally, the generalized DA is given as Algorithm~\ref{alg:generalizedDA}.
\begin{algorithm}[h]
  \SetKwInOut{Input}{input}\Input{$X,\Ch_D,\Ch_H$\quad\textbf{output:} matching $X'\subseteq X$}
  \caption{Generalized DA}\label{alg:generalizedDA}
  $R^{(0)}\ot\emptyset$\;
  \For{$i=1,2,\dots$}{
    $Y^{(i)}\ot \Ch_D(X\setminus R^{(i-1)})$, $Z^{(i)}\ot \Ch_H(Y^{(i)})$\;
    $R^{(i)}\ot R^{(i-1)}\cup(Y^{(i)}\setminus Z^{(i)})$\;
    \lIf{$Y^{(i)}=Z^{(i)}$}{\Return $Y^{(i)}$}
  }
\end{algorithm}
  
\noindent Here, $R^{(i)}$ is a set of rejected contracts at the $i$th iteration. 
Doctors cannot choose contracts in $R^{(i)}$. 
Initially, $R^{(0)}$ is empty. 
Thus, each doctor can choose her most preferred contract. 
The chosen set by doctors is $Y^{(i)}$.
Hospitals then choose $Z^{(i)}$, which is a subset of $Y^{(i)}$. 
If $Y^{(i)}=Z^{(i)}$, i.e., no contract is rejected by the hospitals, 
the mechanism terminates. 
Otherwise, it updates $R^{(i)}$ and repeats the same procedure. 

Hatfield and Milgrom \cite{Hatfield:AER:2005} define 
a notion of stability, which we refer to as \textit{HM-stability}. 
\begin{definition} \textsc{(HM-stability)}
  A matching $X'\subseteq X$ is said to be HM-stable if
  (i) $X'$ satisfies $X'=\Ch_D(X')=\Ch_H(X')$ and 
  (ii) there is no hospital $h$ and set of contracts $X'' \ne \Ch_h(X_h')$ such that $X''=\Ch_h(X_h'\cup X'')\subseteq \Ch_D(X'\cup X'')$.
%
\end{definition}
%
HM-stability unifies stability concepts that are designed for 
each context of (standard) matching problems without constraints. 
Indeed, it implies $B_H$-stability if we require 
the choice functions of hospitals to strictly satisfy budget constraints $B_H$. 
%
Let us next observe the conditions for $\Ch_H$. 
\begin{definition} \textsc{(Substitutability, SUB)}
  For any $X',\allowbreak X''\subseteq X$ with $X''\subseteq X'$,
  $X''\setminus \Ch_H(X'')\subseteq X' \setminus \Ch_H(X')$. 
  Specifically, this condition requires that 
  if contract $x$ is rejected in $X''$, then it is also rejected
  when more contracts are added to~$X''$.
\end{definition}

\begin{definition} \textsc{(Irrelevance of rejected contracts, IRC)}
  For any $X'\subseteq X$ and $X''\subseteq X\setminus X'$,
  $\Ch_H(X')=\Ch_H(X'\cup X'')$ holds if $\Ch_H(X'\cup X'')\subseteq X'$. 
  Thus, this condition requires that, when adding $x$ to $X'$, if
  $x$ is not accepted, then $x$ does not affect the outcomes of other 
  contracts in $X'$.
\end{definition}

\begin{definition} \textsc{(Law of aggregate demand, LAD)}
  For any $X',X''$ with $X''\subseteq X'\subseteq X$, $|\Ch_H(X'')|\leq|\Ch_H(X')|$. 
  In other words, this condition requires that the number of accepted contracts
  increases weakly when more contracts are added. 
\end{definition}
 
Hatfield and Milgrom~\cite{Hatfield:AER:2005} proved that, if $\Ch_H$ satisfies SUB and IRC, 
the generalized DA always produces a matching that is HM-stable. 
If $\Ch_H$ further satisfies LAD, it is strategy-proof for doctors. 


\section{Impossibility and Intractability}
\label{sec:impossibility}
When no hospital is allowed to violate the given constraints, 
stable matchings may not exist.
Let us describe an example in which no stable matching exists~\cite{mongell:el:1986,mcdermid2010,abizada:te:2016}. 
\begin{example}\label{ex:nonexistence}
  Consider a market with three doctors $D=\{d_1,d_2,d_3\}$ and two hospitals $H=\{h_1,h_2\}$. 
  Hospital $h_1$ offers wage $9$, e.g., nine hundred thousand dollars, to doctor $d_1$, and $6$ and $4$ to $d_2$ and $d_3$, 
  while $h_2$ offers $6$ and $4$ to $d_2$ and $d_3$, respectively. Then, the set of offered contracts $X$ is 
  \[
  \{(d_1,h_1,9),(d_2,h_1,6),(d_3,h_1,4),(d_2,h_2,6),(d_3,h_2,4)\}.
  \]
  The doctors' preferences are given as follows: 
  \begin{align*}
    & \succ_{d_1}:~ (d_1,h_1,9), \\
    & \succ_{d_2}:~ (d_2,h_1,6)\succ_{d_2} (d_2,h_2,6),\\ 
    & \succ_{d_3}:~ (d_3,h_2,4)\succ_{d_3} (d_3,h_1,4). 
  \end{align*}
  Next, assume that each hospital $h_i$ simply prefers the set of contracts $X'\subseteq X_{h_i}$ to the set of contracts $X''\subseteq X_{h_i}$ if and only if the total amount of wages in $X'$ is larger than that of $X''$.
  For example, we have $\{ (d_2,h_1,6), (d_3,h_1,4)\} \succ_{h_1} \{(d_1,h_1,9)\} \succ_{h_1} \{(d_2,h_1,6)\}$.
Finally, hospital $h_1$ has a fixed budget of $B_{h_1}=10$ and $h_2$ has $B_{h_2}=6$. We omit the set of contracts beyond these budgets. 

  We claim that there exists no $B_H$-stable matching in this situation. 
  First, assume $d_1$ is assigned to $h_1$ with wage $9$. 
  No more doctors are assigned to this hospital due to the budget constraint. 
  If $d_2$ is assigned to $h_2$, $\{(d_2,h_1,6), (d_3,h_1,4)\}$ is a blocking coalition 
  because $d_2$ prefers $h_1$ to $h_2$, $d_3$ prefers $h_1$ to being unmatched, and 
  $h_1$ prefers $\{(d_2,h_1,6), (d_3,h_1,4)\}$ to $\{(d_1,h_1,9)\}$. 
  If $d_3$ is assigned to $h_2$, $(d_2,h_2,6)$ blocks this matching. 

  Second, assume $d_1$ is not assigned to $h_1$. 
  If $d_2$ and $d_3$ are simultaneously assigned to $h_1$, 
  $d_3$ prefers $h_2$ to $h_1$ and $h_2$ prefers $d_3$ to being unmatched. 
  If they are assigned to different hospitals, respectively, 
  $(d_1,h_1,9)$ blocks this matching regardless of which doctor is assigned to $h_1$. 
  For the remaining cases, since either hospital is being unmatched, 
  some contract or coalition always blocks the matching. 
\end{example}
%
 
The impossibility result raises the issue of the complexity of deciding the 
existence of a $B_H$-stable matching.
McDermid and Manlove~\cite{mcdermid2010} 
considered a special case of our model and proved NP-hardness. 
Hamada \textit{et al.}~\cite{HISY2017} examined a similar model
to ours and Abizada's~\cite{abizada:te:2016} 
and proved that the existence problem is $\Sigma_2^P$-complete.

To deal with the nonexistence of stable matchings, 
we focus on a near-feasible matching that exceeds each budget by a certain amount. 
For each instance of a matching problem, 
our mechanisms find a nearby instance with a stable matching. 
The following theorem implies that, 
to obtain a stable matching,
at least one hospital $h$ needs to increase its budget by nearly $\upperW_h$. 
\begin{theorem}\label{thm:lowerapprox}
  For any positive reals $\alpha<\beta<1$,
  there exists a market $(D,H,X,\succ_D,\succsim_H,B_H)$ such that
  $\upperW_h\le \beta\cdot B_h$ and
  no stable matching exists in any inflated market $(D,H,X,\succ_D, \succsim_H,B_H')$ if $B_h\le B_h'\le (1+\alpha) B_h$ for all $h\in H$. 
\end{theorem}

We place the proof in the appendix. 

\section{New condition: Compatibility}
\label{sec:com}
This section introduces a new condition, 
which we call \textit{compatibility}, on Hatfield and Milgrom's framework 
for a situation in which budget constraints may be violated. 
Let us first consider the following straightforward choice function for each hospital $h$:
\begin{align*}
\Ch_h^*(X')=\max_{\succsim_h}\{X''\subseteq X'_h\mid X''\in\mathcal{X}_h,~w_h(X'')\le B_h\}
\end{align*}
for each $X'\subseteq X$.\footnote{%
When ties occur in the max above, 
we break ties arbitrarily, for example, 
by choosing the lexicographically smallest one with respect to a fixed order of doctors.} 
The generalized DA with the choice function $\Ch_H(X')=\bigcup_{h\in H}\Ch_h^*(X')$ does not always produce a $B_H$-stable matching because such a matching need not exist. 
Furthermore, even if there exists a $B_H$-stable matching, it may produce an unstable matching because the choice function never satisfies the SUB condition. 
In addition, evaluating $\Ch^*_h$ is computationally hard even for additive markets because the problem is equivalent to the well-known \emph{knapsack problem}, which is an NP-hard problem (see e.g., \cite{kellerer2004kp}).
Accordingly, it is neither practical nor reasonable that we require choice functions to satisfy budget constraints. 

What choice function $\Ch_h$ can we construct when we allow it to violate budget constraints? 
SUB, IRC, and LAD can be satisfied and then strategy-proofness is still characterized by them because changing the budgets of hospitals does not affect doctor preferences. 
However, SUB and IRC are not sufficient to admit a stable matching in our sense. 
Intuitively, 
the choice function should select the best set of contracts under certain budgets 
among the doctors offered to each hospital.
Even if it satisfies SUB and LAD, the best set is not always selected via the generalized DA process. 
Therefore, a hospital with non-optimal utility can form a blocking coalition to improve its utility. 
To prevent this, we introduce a novel condition: 
\begin{definition} \textsc{(Compatibility, COM)}
  Consider a hospital $h$ with a weak preference $\succsim_h$, a budget $B_h$, and contracts $X_h$. 
  For any $X''\subseteq X'\subseteq X_h$ such that $w_h(X'')\le \max\{B_h,w_h(\Ch_h(X'))\}$ and $X''\in\mathcal{X}_h$,
  it holds that
  \[
  \Ch_h(X')\succsim X''. 
  \]
\end{definition}

With this condition, the output of the choice function $\Ch_h$ is guaranteed to be
the optimal set of contracts under a certain capacity that is greater than or equal to the predefined capacity. 

We next prove that COM together with SUB and IRC characterizes stable matchings 
when budget constraints may be violated. 
\begin{theorem}\label{thm:compatible2stable}
  Suppose that for every hospital the choice function satisfies SUB, IRC, and COM. 
  The generalized DA produces a matching $X'$ that is $B'_H$-stable, 
  where $B'_H=(\max\{B_h,w_h(X')\})_{h\in H}$. 
\end{theorem}

\begin{proof}
Let the mechanism terminate at the $l$th iteration, i.e., $X'=Y^{(l)}=Z^{(l)}$.
From its definition, it is immediately derived that 
the union of $Y^{(i)}$ and $R^{(i)}$ is nondecreasing in $i(\leq l)$; i.e., 
for any $i\in\{2,3,\dots,l\}$, 
\begin{align}\label{eq:T}
  Y^{(i)}\cup R^{(i)}\supseteq Y^{(i-1)}\cup R^{(i-1)}. 
\end{align}
%
For notational simplicity, we refer to $Y^{(i)}\cup R^{(i)}$ as $T^{(i)}$. 

Next, to obtain $\Ch_H(X'\cup R^{(l)})=X'$, 
we claim that $\Ch_H(T^{(i)})=Z^{(i)}$ for any $i\in\{1,2,\dots,l\}$ by induction.
For the base case $i=1$, we have $\Ch_H(T^{(1)})=\Ch_H(Y^{(1)})=Z^{(1)}$ 
since $R^{(1)} =
Y^{(1)}\setminus Z^{(1)} \subseteq Y^{(1)}$. 
For the general case $i>1$, we suppose $\Ch_H(T^{(i-1)})=Z^{(i-1)}$.
From~\eqref{eq:T} and the SUB condition, we have
$T^{(i-1)}\setminus\Ch_H(T^{(i-1)})\subseteq T^{(i)}\setminus\Ch_H(T^{(i)})$.
By the inductive hypothesis, we transform the left term $T^{(i-1)}\setminus\Ch_H(T^{(i-1)})$
to 
\begin{align*}
  T^{(i-1)}\setminus\Ch_H(T^{(i-1)})
  &=(Y^{(i-1)}\cup R^{(i-1)})\setminus Z^{(i-1)}
  =(R^{(i-1)}\setminus Z^{(i-1)})\cup (Y^{(i-1)}\setminus Z^{(i-1)})\\
  &= R^{(i-1)}\cup (Y^{(i-1)}\setminus Z^{(i-1)}) = R^{(i-1)}.
\end{align*}
Hence, it holds that $R^{(i-1)}\subseteq T^{(i)}\setminus\Ch_H(T^{(i)})$
and thus $\Ch_H(T^{(i)})$ includes no contract in $R^{(i-1)}$.
Together with the IRC condition, $\Ch_H(T^{(i)})$ is equal to
\begin{align*}
  \Ch_H(T^{(i)}\setminus R^{(i-1)}) 
  &= \Ch_H((Y^{(i)}\cup R^{(i)})\setminus R^{(i-1)})\\
  &= \Ch_H((Y^{(i)}\cup R^{(i-1)})\setminus R^{(i-1)}) = \Ch_H(Y^{(i)}) = Z^{(i)}. 
\end{align*}
The second equality holds because 
$Y^{(i)}\cup R^{(i)}=Y^{(i)}\cup (R^{(i-1)}\cup(Y^{(i)}\setminus Z^{(i)}))=Y^{(i)}\cup R^{(i-1)}$.
Consequently, we obtain the claim and, since $X'=Y^{(l)}=Z^{(l)}$, we get $\Ch_H(X'\cup R^{(l)})=X'$. 

Next, since $\Ch_h$ is COM, 
$\Ch_h(X_h'\cup R_h^{(l)})\succsim_h X''$ holds for any $X''\subseteq X_h'\cup R_h^{(l)}$ such that $w_h(X'')\le \max\{B_h,w_h(\Ch_h(X_h'\cup R_h^{(l)}))\}$.
Here, as $\Ch_h(X_h'\cup R_h^{(l)})=X_h'$,
$\Ch_h(X_h')\succsim_h X''$ holds for any $X''\subseteq X_h'\cup R_h^{(l)}$ such that $w_h(X'')\le B_h'$.


Suppose, contrary to our claim, that $X''\subseteq X_h$ is a blocking coalition for a hospital $h$.
By the definition of a blocking coalition, 
(i) $X''\succeq_{x_D}X'$ for all $x\in X''$, (ii) $X'' \succ_h X_h'$, and (iii) $w_h(X'')\le B_h'$.
The condition (i) implies $x\in X'\cup R^{(l)}$ for all $x\in X''$ and, hence, $X''\subseteq X_h'\cup R_h^{(l)}$ holds.
This contradicts the fact that $\Ch_H(X'\cup R^{(l)})=X'$ for any $X'$. 
Therefore, $X'$ is $B'_H$-stable where $B'_H=(\max\{B_h,w_h(X')\})_{h\in H}$. 
\end{proof}

Note that this theorem does not specify how much a hospital may exceed its budget. 
Here, one can define a choice function such that the hospital can afford to hire
all of the doctors who have accepted its contracts
far beyond the predefined budgets. 
The theorem simply ensures that if a choice function satisfies COM, in addition to SUB and IRC, 
the generalized DA admits a $B'_H$-stable matching $X'$ with $B'_h=\max\{B_h, w_h(X')\}$ for each hospital.
 
We also point out that, if each hospital $h$ knows the selectable contracts, i.e., $Y_h^{(l)}\cup R_h^{(l)}$, in advance, 
it only needs to select $\max_{\succsim_h}\{X''\subseteq X_h'\cup R_h^{(l)}\mid w_h(X'')\le B_h'\}$ 
for a certain budget $B_h'~(\ge B_h)$. 
However, the selectable contracts are difficult to predict because the resulting set depends on the choice function itself.
Designing or finding a choice function that satisfies the required conditions and
only violates budget constraints to an acceptable extent is not straightforward.
%


\section{Near-Feasible Stable Mechanisms}
\label{sec:mechs}

In matching with constraints~\cite{kamada:aer:2015,Goto:aij:2016,kurata:jair:2017}, 
designing a desirable 
mechanism essentially tailors choice functions for hospitals to satisfy necessary conditions and constraints simultaneously. Although this is still a challenging task, we manage to construct choice functions as an analog to approximation or online algorithms for knapsack problems, guided by our analysis. 

Let us start from Dantzig's greedy algorithm for \textit{fractional} knapsack problems~\cite{dantzig1957}, assuming the utilities of each hospital are additive.  
This algorithm greedily selects contracts with respect to \textit{utility per unit wage} and then outputs an optimal but fractional solution. 
We need to develop an algorithm that always provides an integral solution. 
Roughly speaking,
we must provide an algorithm (choice function) that satisfies the necessary conditions, e.g., SUB and COM, 
for any set of contracts $X'$ given at each round of the generalized DA. 
At the same time,
we need to let the algorithm determine how much budget should be exceeded beyond the predefined amount 
(how many contracts should be chosen). 
Indeed, at each round, it is difficult to predict the amount of over-budget excess 
without violating the necessary conditions. 
In what follows, we propose two choice functions that adaptively specify 
how much budgets should be spent within the generalized DA process. 



\subsection{Strategy-Proof Stable Mechanism}
This subsection proposes a strategy-proof mechanism that outputs a matching $X'$
that is $B'_H$-stable, where $B'_h$ is at most $\upperW_h\cdot \lceil B_h/\lowerW_h\rceil$ for any $h\in H$. 
Let $k_h=\lceil B_h/\lowerW_h\rceil$. 
The choice function greedily takes the top $\min\{k_h, |X'|\}$ contracts according to the utility per unit wage. 
%
Formally, it is given as Algorithm~\ref{alg:spchoice}.
\begin{algorithm}[h]
  \SetKwInOut{Input}{input}\Input{$X'\subseteq X_h$\quad\textbf{output:} $\Ch_h(X')$}
  \caption{}\label{alg:spchoice}
  Initialize $Y\ot\emptyset$\; 
  Sort $X'$ in descending order of utility per unit wage\;
  \For{$i=1,2,\dots,\min\{k_h,|X'|\}$}{
    add the $i$th contract in $X'$ to $Y$\;
  }
  \Return $Y$\;
\end{algorithm}
%

We can implement the mechanism to run in $O(|X|\log |X|)$ time by using heaps.
Let us prepare a min-heap with respect to each hospital's utility per unit wage. 
We derive the time complexity from an ``amortized'' analysis and 
can add a newly proposed contract $x\in X_h$ to the heap for $h$ in $O(\log(|X_h|))$ time. 
When a hospital $h$ rejects a contract, we can delete it in $O(\log(|X_h|))$ time. 
Hence, the total time complexity is $O(\sum_{h\in H} |X_h| \log(|X_h|)) = O(|X|\log(|X|))$.

Let us illustrate this mechanism via an example. 
\begin{example}
  \label{ex:mechanism 1}
  Consider a market with five doctors $D=\{d_1,d_2,d_3,d_4,d_5\}$ and 
  two hospitals $H=\{h_1,h_2\}$. The set of offered contracts is 
  {\small
    \begin{align*}\notag
      X=\{ & (d_1,h_1,57), (d_2,h_1,50), (d_3,h_1,42), (d_4,h_1,55), (d_5,h_1,50), \\
      &(d_1,h_2,100),(d_2,h_2,100),(d_3,h_2,100), (d_4,h_2,100), (d_5, h_2,100) \}. 
    \end{align*}}
  Here, $h_1$ offers the doctors wages from $42$ to $57$, while 
  $h_2$ offers each of them wage $100$. 
  We assume that the preferences of the doctors are 
  {\small
    \begin{align*}\notag
      & \succ_{d_1}:~ (d_1,h_1,57)  \succ_{d_1} (d_1,h_2,100),\\
      & \succ_{d_2}:~ (d_2,h_1,50)  \succ_{d_2} (d_2,h_2,100), \\\notag
      & \succ_{d_3}:~ (d_3,h_1,42)  \succ_{d_3} (d_3,h_2,100),\\
      & \succ_{d_4}:~ (d_4,h_1,55)  \succ_{d_4} (d_4,h_2,100), \\\notag
      & \succ_{d_5}:~ (d_5,h_2,100) \succ_{d_5} (d_5,h_1,50).  \notag
    \end{align*}}
  The utilities of the hospitals are given in Table~\ref{tab:ex:1}.
  {\begin{table}[tb]
      \centering
      \caption{Utilities of hospitals and utilities per unit wage.}
      \label{tab:ex:1}
        \begin{tabular}[htbp]{|c|c|c||c|c|c|}
          $x \in X_{h_1}$  & $f_{h_1}$ & $f_{h_1}/w$  & $x \in X_{h_2}$ & $f_{h_2}$ & $f_{h_2}/w$ \\\hline
          $(d_1,h_1,57)$   & $111$     & $1.95$      & $(d_1,h_2,100)$ & $50$      & $0.50$ \\
          $(d_2,h_1,50)$   & $98$      & $1.96$      & $(d_2,h_2,100)$ & $30$      & $0.30$ \\
          $(d_3,h_1,42)$   & $83$      & $1.98$      & $(d_3,h_2,100)$ & $20$      & $0.20$ \\ 
          $(d_4,h_1,55)$   & $110$     & $2.00$      & $(d_4,h_2,100)$ & $10$      & $0.10$ \\
          $(d_5,h_1,50)$   & $101$     & $2.02$      & $(d_5,h_2,100)$ & $40$      & $0.40$ 
        \end{tabular}
    \end{table}}
  Each hospital has a common fixed budget $100$ ($B_{h_1}=B_{h_2}=100$). 
  
  First, each doctor chooses her most preferred contract: 
  \begin{align*}
    X' = \{&(d_1,h_1,57), (d_2,h_1,50), (d_3,h_1,42), (d_4,h_1,55), (d_5,h_2,100)\}.
  \end{align*}
  Since $\lceil B_{h_1}/\lowerW_{h_1} \rceil = 3$, 
  $\Ch_{h_1}(X')$ chooses the top three contracts according to the ranking of utilities per unit wage shown in Table~\ref{tab:ex:1},
  which is, $\{(d_4,h_1,55), (d_3,h_1,42), (d_2,h_1,50) \}$. 
  In addition, 
  $\Ch_{h_2}(X')$ chooses the top $\lceil B_{h_2}/\lowerW_{h_2} \rceil = 1$ contract, that is, $\{(d_5,h_2,100)\}$.
  
  Then, $d_1$ chooses her second preferred contract, 
  \begin{align*}
      X' = \{& (d_1,h_2,100), (d_2,h_1,50), (d_3,h_1,42), (d_4,h_1,55), (d_5,h_2,100)\}. 
  \end{align*}
  $\Ch_{h_2}(X')$ is $\{(d_1,h_2,100)\}$, whose utility per unit wage is larger than $(d_5,h_2,100)$. 
  
  Next, $d_5$ chooses her second preferred contract, i.e., $(d_5,h_1,50)$, whose utility per unit wage is $2.02$. 
  Since this is higher than the other contract in $X'_{h_1}$, $(d_2,h_1,50)$ is rejected. 
  Thus, $d_2$ chooses her second preferred contract, $(d_2,h_2,100)$: 
  \begin{align*}
           X'= \{&(d_1,h_2,100), (d_2,h_2,100), (d_3,h_1,42), (d_4,h_1,55), (d_5,h_1,50)\}. 
  \end{align*}
  $\Ch_{h_2}(X')=\{(d_1,h_2,100)\}$, since it 
  has a higher utility per unit wage than $(d_2,h_2,100)$. 
  Finally, since $d_2$ no longer has a preferred contract, 
  {\small\begin{align*}
      X'= \{(d_1,h_2,100), (d_3,h_1,42), (d_4,h_1,55), (d_5,h_1,50)\}.   
    \end{align*}}
  No contract is rejected and the mechanism terminates. 
\end{example}

The choice function satisfies the following conditions: 
\begin{lemma}\label{lem:spchprop} 
  For each hospital $h$, 
  the choice function defined in Algorithm~\ref{alg:spchoice} 
  is SUB, IRC, LAD, and COM. 
\end{lemma}

\begin{proof}
It is straightforward that the choice function is SUB, IRC, and LAD
because it simply picks at most the top $\min\{k_h,|X'|\}$ contracts. 
Next, let us turn to COM. Let $X''\subseteq X'\subseteq X_h$. 
If $|X'|\le k_h$, since the choice function picks all contracts in $X'$, 
$f_h(\Ch_h(X'))=f_h(X')\ge f_h(X'')$ clearly holds. 
On the other hand, if $|X'|>k_h$, since it picks $k_h$ contracts, 
we have $w_h( \Ch_h(X'))\geq \lowerW_h \cdot k_h \geq B_h$. 
Thus, it is sufficient to claim that 
\[
f_h(\Ch_h(X'))\geq f_h(X'')\ \text{if}\ w_h(X'')\leq w_h(\Ch_h(X'))
\]
for any $X''\subseteq X' \subseteq X$. 

Since the choice function greedily pick $k_h$ contracts with respect to the utility per unit wage, 
the chosen contracts yield the optimal utility of a fractional knapsack problem~\cite{dantzig1957}. 
Also, to maximize the utility of hospital $h$ with $X''$, we need to solve an integral knapsack problem. 
Therefore, $f_h(\Ch_h(X'))$ is at least

\begin{align*}
&\max_{z\in [0,1]^{X'} } \left\{ \sum_{x\in X'} f_h(x)\cdot z_x \mid \sum_{x\in X'}x_W\cdot z_x \leq w_h(\Ch_h(X'))\right\} \\
\geq &\displaystyle\max_{z\in \{0,1\}^{X'}} \left\{ \sum_{x\in X'} f_h(x)\cdot z_x \mid \sum_{x\in X'}x_W\cdot z_x \leq w_h(\Ch_h(X'))\right\} \\
= &\displaystyle\max_{Y\subseteq X'} \left\{ \sum_{x\in Y} f_h(x) \mid \sum_{x\in Y} x_W \leq w_h(\Ch_h(X')) \right\} \geq f_h(X'').
\end{align*}
Note that the first inequality is derived from the fact that the optimal value of the fractional knapsack problem 
is never worse than that of the integral one. 
Thus, the choice function $\Ch_h$ satisfies COM and the proof is complete. 
%
\end{proof}

Next, we show the upper bound of the increment of the budgets.
The following lemma clearly holds from 
$|\Ch_h(X')|\le k_h$ and $x_W\le \upperW_h$ for all $x\in X_h$.
\begin{lemma}\label{lem:spchapprox}
  For each choice function defined in Algorithm~\ref{alg:spchoice} and
  a set of contracts $X'\subseteq X_h$, it holds that 
  \begin{align*}
    w_h(\Ch_h(X'))\le \upperW_h\cdot k_h~(=\upperW_h\cdot\lceil B_h/\lowerW_h\rceil).
  \end{align*}
\end{lemma}

We summarize the above arguments in the following theorem: 
\begin{theorem}
  The generalized DA mechanism with the choice functions defined in Algorithm~\ref{alg:spchoice} is
  strategy-proof for doctors and produces a $B_H'$-stable matching such that 
  $B_h\le B_h'\le \upperW_h\cdot k_h$ for any $h\in H$. 
  In addition, the mechanism can be implemented to run in $O(|X|\log|X|)$ time. 
\end{theorem}


\begin{modified}
Finally, we note that 
this mechanism is \textit{almost tight} 
as long as we use the choice functions that satisfy LAD and COM. 


\begin{theorem}
  For any $\upperW_h$, $\lowerW_h$, and $B_h$ ($0<\lowerW_h\le\upperW_h\le B_h$),
  there exists a set of contracts $X_h$ and an additive utility function $f_h: X_h\to\mathbb{R}_{+}$
  such that any choice function $\Ch_h: 2^{X_h}\to 2^{X_h}$ satisfies 
  $\Ch_h(X')>\upperW_h\cdot (B_h-\upperW_h)/\lowerW_h$ for some $X'\subseteq X_h$
  if $\Ch_h$ is LAD and COM.
\end{theorem}
\begin{proof}
  Let $k=\lfloor B_h/\lowerW_h\rfloor$, $X_h=\{x_1,\dots,x_{2k}\}$, and
  \begin{align*}
    x_i=\begin{cases}
      (d_i,h,\lowerW_h)&(i=1,\dots,k),\\
      (d_i,h,\upperW_h)&(i=k+1,\dots,2k),
    \end{cases}
    \quad
    f_h(x_i)=\begin{cases}
      \lowerW_h&(i=1,\dots,k),\\
      2\cdot\upperW_h&(i=k+1,\dots,2k).
    \end{cases}
  \end{align*}
  Since $\Ch_h$ is COM, 
  $\Ch_h(\{x_{1},\dots,x_{k}\})=\{x_{1},\dots,x_{k}\}$ holds.
  Thus, we have $|\Ch_h(X_h')|\ge k$ if $\{x_{1},\dots,x_{k}\}\subseteq X_h'\subseteq X_h$
  because $\Ch_h$ is LAD.
  Let $|\Ch(X_h)\cap\{x_1,\dots,x_k\}|=s$ and $|\Ch(X_h)\cap\{x_{k+1},\dots,x_{2k}\}|=t$.
  Here, $s+t\ge k$ holds.
  Without loss of generality, we may assume that 
  $\Ch(X_h)=\{x_1,\dots,x_s,x_{k+1},\dots,x_{k+t}\}$.
  
  If $s\ge\upperW_h/\lowerW_h$ and $t<k$,
  $\Ch_h$ is not COM 
  as 
  $f_h(\Ch(X_h))<f_h(X')$ and $w_h(\Ch_h(X_h))\ge w_h(X')$
  for $X'=\{x_{\lceil \upperW_h/\lowerW_h\rceil+1},\dots,x_{s},x_{k+1},\dots,x_{k+t}\}$.
  Thus, we have $s<\upperW_h/\lowerW_h$ or $t=k$.
  If $t=k$, we obtain that
  \begin{align*}
    w_h(\Ch_h(X_h))=\lowerW_h\cdot s+\upperW_h\cdot t\ge \upperW_h\cdot k=\upperW_h\cdot\left\lfloor\frac{B_h}{\lowerW_h}\right\rfloor>
    \upperW_h\cdot \left(\frac{B_h}{\lowerW_h}-1\right)
    \ge\upperW_h\cdot \frac{B_h-\upperW_h}{\lowerW_h}.
  \end{align*}
  On the other hand, if $s<\upperW_h/\lowerW_h$, we get 
  \begin{align*}
    w_h(\Ch_h(X_h))
    &=\lowerW_h\cdot s+\upperW_h\cdot t\\
    &\ge \upperW_h\cdot k-(\upperW_h-\lowerW_h)\cdot s\\
    &> \upperW_h\cdot \left\lfloor\frac{B_h}{\lowerW_h}\right\rfloor-(\upperW_h-\lowerW_h)\cdot \frac{\upperW_h}{\lowerW_h}\\
    &\ge \upperW_h\cdot \frac{B_h-\upperW_h}{\lowerW_h}. \qedhere
  \end{align*}
\end{proof}
\end{modified}

\subsection{Non-Strategy-Proof Stable Mechanism}
This subsection proposes a stable mechanism that is not strategy-proof, but 
improves the budget bound, that is, this mechanism outputs 
a matching $X'$ that is 
$B'_H$-stable where $B'_h$ is at most $B_h+\upperW_h$ for any $h\in H$. 
This is the best possible bound from Theorem~\ref{thm:lowerapprox}.

As with the first, the second choice function greedily picks the top 
$\min \{k_h, |X'|\}$ contracts. However,
$k_h$ is defined as $\min\bigl\{ k \mid \sum_{i=1}^{k}x^{(i)} \geq B_h \bigr\}$,
where $x^{(i)}$ denotes the $i$th highest contract with respect to the utility per unit wage. 
%
Formally, it is given as Algorithm~\ref{alg:nspchoice}.

\begin{algorithm}
\SetKwInOut{Input}{input}\Input{$X'\subseteq X_h$\quad\textbf{output:} $\Ch_h(X')$}
\caption{}\label{alg:nspchoice}
Initialize $Y\ot \emptyset$\;
Sort $X'$ in descending order of utility per unit wage\;
\For{$i=1,2,\dots,|X'|$}{
  let $x$ be the $i$th contract in $X'$\;
  \lIf{$w_h(Y) < B_h$}{$Y\ot Y\cup\{x\}$}
}
\Return $Y$\;
\end{algorithm}
Note that the running time is $O(|X|\log|X|)$, as with the first mechanism. 
%
\begin{modified}
Let us illustrate this mechanism via an example. 
\begin{example}
  \label{ex:mechanism 2}
  We consider a situation that is identical to Example~\ref{ex:mechanism 1} with the first two rounds the same. 
  
  
  At the third round, when $d_5$ chooses $(d_5,h_1,50)$, 
  \begin{align*}
      X' = \{&(d_1,h_2,100), (d_2,h_1,50), (d_3,h_1,42), (d_4,h_1,55), (d_5,h_1,50)\}. 
  \end{align*}
  The number of contracts $\Ch_{h_1}(X')$ chooses changes from three to two; 
  the total wage of the first two contracts for $h_1$ ($105$) exceeds the budget limit of $100$. 
  Thus, $(d_2,h_1,50)$ and $(d_3,h_1,42)$ are rejected. 

  Next, $d_2$ and $d_3$ choose their second preferred contracts, i.e., $(d_2,h_2,100)$ and $(d_3,h_2,100)$, 
  but these contracts are also rejected. Finally, since they no longer have preferred contracts, 
  \begin{align*}
    X'=\{(d_1,h_2,100), (d_4,h_1,55), (d_5,h_1,50)\}.     
  \end{align*}
  No contract is rejected and the mechanism terminates. 
\end{example}
\end{modified}
We show the conditions that this mechanism satisfies below. 
\begin{lemma}\label{lem:nspchprop}
  For each hospital, the choice function defined in Algorithm~\ref{alg:nspchoice} is
  SUB, IRC, and COM. 
\end{lemma}

\begin{proof}
IRC clearly follows from the definition of the choice functions. 
Next, we claim that the choice functions satisfy SUB. 
Let $X''\subseteq X'\subseteq X_h$. 
By definition, 
the utility per unit wage of any contract in $\Ch_h(X')~(\supseteq \Ch_h(X')\cap X'')$
is higher than
that of any contract in $X'\setminus \Ch_h(X')~(\supseteq X''\setminus \Ch_h(X'))$. 
Hence, we can partition $X''$ into two subsets: 
$H=\Ch_h(X')\cap X''$ and $L=X''\setminus \Ch_h(X')$. 
Any contract in $H$ has a higher utility per unit wage than any contract in $L$. 
When $\Ch_h$ takes $X''$ as an input,
it first picks all of the contracts in $H$ and some contracts in $L$. 
Therefore, we obtain $\Ch_h(X')\cap X''\subseteq \Ch_h(X'')$ 
and derive the SUB condition: 
\begin{align*}
  X''\setminus \Ch(X'') & \subseteq X''\setminus(\Ch_h(X')\cap X'') \subseteq X'\setminus\Ch_h(X'). 
\end{align*}

Finally, we prove COM. 
Let $X'=\{x^{(1)},\dots,x^{(|X'|)}\}\subseteq X_h$, 
where the contracts are arranged in decreasing order of the utility per unit wage. 
If $w_h(X')\le B_h$, 
then it is clear that $\Ch_h(X')=X'$ and $f_h(\Ch_h(X'))\ge f_h(X'')$ hold for any $X''\subseteq X'$.
Otherwise, 
let $\Ch_h(X')=\{x^{(1)},\dots,x^{(k)}\}$. 
Here, \[w_h(\{x^{(1)},\dots,x^{(k-1)}\})< B_h\le w_h(\{x^{(1)},\dots,x^{(k)}\})\] holds. 
As described in Lemma~\ref{lem:spchprop}, 
since the greedy solution $\Ch_h(X')$ is optimal, 
we have 
$f_h(\Ch_h(X'))\ge f_h(X'')$
for any $X''\subseteq X'$ such that $w_h(X'')\le w_h(\Ch_h(X'))$.
Thus, the lemma holds.
\end{proof}
Demonstrating that Algorithm~\ref{alg:nspchoice} does not satisfy LAD is 
straightforward. In Example \ref{ex:mechanism 1}, 
when a set of contracts $\{(d_2, h_1, 50), (d_3, h_1, 42), (d_4, h_1, 55)\}$ is given,
the choice function chooses all three contracts. 
Here, if $(d_1,h_1,57)$ is further added, it chooses only two contracts, 
specifically, $\{(d_1, h_1, 57), (d_2, h_1, 50)\}$.
Thus, the second mechanism fails to satisfy LAD.

We next show an upper bound of the increment of the budgets. 
Since our choice function chooses a set of contracts 
that exceeds the capacity by at most one contract,
we have the following lemma: 
\begin{lemma}\label{lem:nspchapprox}
  For each choice function defined in Algorithm~\ref{alg:nspchoice} 
  and a set of contracts $X'\subseteq X_h$, it holds that
  \begin{align*}
    w_h(\Ch_h(X'))< B_h+\upperW_h.
  \end{align*}
\end{lemma}
\begin{proof}
  Sort contracts in $X'$ in descending order of utility per unit wage. 
  We obtain $X'=\{x^{(1)},\dots,x^{(|X'|)}\}\subseteq X_h$. 
  %
  If $w_h(X')\le B_h$, then it is clear that $\Ch_h(X')=X'$ and
  $w_h(\Ch_h(X'))=w_h(X')\le B_h\le B_h+\upperW_h$. 
  Otherwise, 
  when $\Ch_h(X')$ chooses $k$ contracts, i.e., $\Ch_h(X')=\{x^{(1)},\dots,x^{(k)}\}$, 
  {\small\[ w_h(\{x^{(1)},\dots,x^{(k-1)}\})< B_h\le w_h(\{x^{(1)},\dots,x^{(k)}\})\]}
  holds. Then we obtain
  {\small\begin{align*}
    w_h(\Ch_h(X')) & = w_h(\{x^{(1)},\dots,x^{(k)}\}) \\
    & =w_h(\{x^{(1)},\dots,x^{(k-1)}\})+w_h(\{x^{(k)}\}) \\
    & < B_h+\upperW_h, 
  \end{align*}}
  which proves the lemma.
\end{proof}

We summarize the results for our second mechanism follows: 
\begin{theorem}
  The generalized DA mechanism with the choice functions defined in Algorithm~\ref{alg:nspchoice} produces a set of contracts $X'$ 
  that is $B'_H$-stable where $B_h\le B_h'< B_h+\upperW_h$ for any $h\in H$.
  In addition, the mechanism can be implemented to run in $O(|X|\log|X|)$ time. 
\end{theorem}


\begin{modified}
The following confirms that the mechanism is not strategy-proof. 
Consider a market with three doctors $D=\{d_1,d_2,d_3\}$ and two hospitals $H=\{h_1,h_2\}$. 
The set of offered contracts is 
\begin{align*}
X=\{(d_1,h_1,1),(d_1,h_2,1),(d_2,h_1,2),(d_2,h_2,1),(d_3,h_1,1),(d_1,h_3,1)\}.
\end{align*}
The doctors' preferences are given as follows: 
\begin{align*}
  & \succ_{d_1}:~ (d_1,h_1,1)\succ_{d_1} (d_1,h_2,1),\\
  & \succ_{d_2}:~ (d_2,h_2,1)\succ_{d_2} (d_2,h_1,2),\\
  & \succ_{d_3}:~ (d_3,h_2,1)\succ_{d_3} (d_3,h_1,1). 
\end{align*}
The utility of each contract is given as follows:
\begin{align*}
f_{h_1}&(\{(d_1,h_1,1)\})=1,\ f_{h_1}(\{(d_2,h_1,2)\})=10,\ f_{h_1}(\{(d_3,h_1,1)\})=1,\\
f_{h_2}&(\{(d_1,h_2,1)\})=3,\ f_{h_2}(\{(d_2,h_2,1)\})=1,\ \text{and}\ f_{h_2}(\{(d_3,h_2,1)\})=2.
\end{align*}
The budgets of hospitals $h_1$ and $h_2$ are $B_{h_1}=2$ and $B_{h_2}=1$, respectively.

For this market, our mechanism outputs a stable-matching $X'=\{(d_1,h_2,1),(d_2,h_1,2)\}$.
On the other hand, if $d_3$ misreports his or her preference as $(d_3,h_1,1)\succ_{d_3}' (d_3,h_2,1)$,
the outcome of the mechanism is $X''=\{(d_1,h_1,1),\allowbreak (d_2,h_2,1),(d_3,h_1,1)\}$.
Thus, the mechanism is not strategy-proof since $d_3$ has an incentive to misreport.
\end{modified}


\subsection{Non-Existence of Doctor-Optimal Stable Matchings}
\label{sec:doctor-optimality}

It is known that a doctor-optimal stable matching often fails to exist 
in matchings with distributional constraints~\cite{kamada:aer:2015}. 
Unfortunately, this also holds in our matching problem, 
even if there is a stable matching that exactly satisfies the budget constraints. 
To illustrate this, let us consider a market with four doctors $D=\{d_1,d_2,d_3,d_4\}$ and two hospitals $H=\{h_1,h_2\}$. 
The set of offered contracts is 
\[ X=\{(d_1,h_1,1), (d_2,h_1,2), (d_3,h_1,1), (d_4,h_1,1), (d_3,h_2,1),(d_4,h_2,1)\}.\]
Suppose that the doctors' preferences are 
\begin{align*}\notag
  & \succ_{d_1}:~ (d_1,h_1,1),\\
  & \succ_{d_2}:~ (d_2,h_1,2),\\
  & \succ_{d_3}:~ (d_3,h_1,1)\succ_{d_3}(d_3,h_2,1),\ \mathrm{and}\\
  & \succ_{d_4}:~ (d_4,h_2,1)\succ_{d_4}(d_4,h_1,1).
\end{align*}
The utility of each contract is 
\begin{align*}
f_{h_1}&(\{(d_1,h_1,1)\})=7,\ f_{h_1}(\{(d_2,h_1,2)\})=6,\ f_{h_1}(\{(d_3,h_1,1)\})=1, f_{h_1}(\{(d_4,h_1,1)\})=4, \\
f_{h_2}&(\{(d_3,h_2,1)\})=2,\ \text{and}\ f_{h_2}(\{(d_4,h_2,1)\})=1. 
\end{align*}
The budgets of hospitals $h_1$ and $h_2$ are $B_{h_1}=2$ and $B_{h_2}=1$, respectively.

For this market, our mechanisms output a stable matching $X'=\{(d_1,h_1,1),(d_3,h_2,1), (d_4,h_1,1)\}$
that exactly satisfies the budget constraints. The matching clearly improves without hurting anyone 
if $d_3$ and $d_4$ exchange their positions. Since $d_3$ prefers $h_1$ to $h_2$ and $d_4$ prefers $h_2$ to $h_1$,
they unanimously prefer another stable matching $X''=\{(d_1,h_1,1), (d_3,h_1,1),\allowbreak (d_4,h_2,1)\}$.  

\begin{modified}
\subsection{Two Special Cases}
\label{sec:special cases}
This section examines two special hospital utilities cases. 
%
First, we assume that each hospital has utility over a set of contracts 
that is proportional to the total amount of wages. 
Formally, for every $h\in H$, $X'\in X_h$, and a constant $\gamma_h(>0)$, 
\[f_h(X')=\gamma_h\cdot w_h(X')\]
holds. In this case, we can make the second mechanism strategy-proof without sacrificing the budget bound
although stable matching may not exist as in Example \ref{ex:nonexistence}.
Specifically, we modify Algorithm~\ref{alg:nspchoice} to 
i) sort $X'$ in increasing order of wage, instead of decreasing order of utility per unit wage; %
ii) pick the contracts in order while keeping the total wage less than $B_h$; and 
iii) add the contract with the highest wage unless it is already chosen. 
Formally, we define this as Algorithm \ref{alg:proputil}. 

\begin{algorithm}
\SetKwInOut{Input}{input}\Input{$X'\subseteq X_h$\quad\textbf{output:} $\Ch_h(X')$}
\caption{}\label{alg:proputil}
Initialize $Y\ot \emptyset$\;
Sort $X'$ in increasing order of wages\; %
\For{$i=1,2,\dots,|X'|-1$}{
  let $x$ be the $i$th contract in $X'$\;
  \lIf{$w_h(Y\cup\{x\}) < B_h$}{$Y\ot Y\cup\{x\}$}
}
add the $|X'|$th contract (highest wage contract) to $Y$\;
\Return $Y$\;
\end{algorithm}

\begin{lemma}\label{lem:proputil}
  For each hospital $h$, the choice function defined in 
  Algorithm~\ref{alg:proputil} is SUB, IRC, LAD, and COM. 
\end{lemma}
\begin{proof}
  Let $X'$ be $\{x^{(1)},\dots,x^{(|X'|)}\}\subseteq X_h$ such that 
  $x_{W}^{(1)}<\dots<x_{W}^{(|X'|)}$ and 
  let $k$ be the largest integer that satisfies $\sum_{i=1}^{k}x_W^{(i)}< B_h$ and $k\le |X'|-1$.
  Also, let $X''$ be a subset of $X'$ that contains $s~(=|X''|)$ contracts, that is, 
  $X''=\{x^{(j_1)},\dots,x^{(j_s)}\}$ such that $1\le j_1<\dots<j_s\le |X'|$.
  Let $k'$ be the largest integer that satisfies $\sum_{i=1}^{k'}x_W^{(j_i)}< B_h$ and $k'\le s-1$.
  Then, we have $k'\le k\le j_{k'}$ by $x_W^{(i)}\le x_W^{(j_i)}$ for all $i$. 

  From this setup, we obtain 
  \begin{align*}
    \Ch_h (X')\cap X''
          & = \{x^{(1)},\dots,x^{(k)},x^{(|X'|)}\} \cap\{x^{(j_1)},\dots,x^{(j_s)}\}\\
          & \subseteq \{x^{(j_1)},\dots,x^{(j_{k'})},x^{(j_s)}\}=\Ch_h(X''), 
  \end{align*}
  which is equivalent to the SUB condition. We also obtain 
  \[
  |\Ch_h(X'')| = k'+1 \le  k+1 = |\Ch_h(X')|, 
  \]
  which implies that $\Ch_h$ is LAD. 
  Furthermore, the fact that the choice function for $h$ satisfies SUB and LAD simultaneously implies IRC. 

  Finally, we prove COM. 
  First, let us consider a matching $X'$ whose total wages for hospital $h$ is less than or equal to $B_h$, 
  i.e., $w_h(X')\le B_h$. In this case, since the choice function picks all contracts in $X'$, 
  we have 
  \begin{align*}
    f_h(\Ch_h(X')) = f_h(X') 
                   = \max\{f_h(X''')\mid X'''\subseteq X',~w_h(X''')\le B_h\}.
  \end{align*}
  Second, let us consider a matching $X'$ whose total wages for hospital $h$ is greater than $B_h$, 
  i.e., $w_h(X')>B_h$.
  In this case, $w_h(\Ch_h(X'))\ge B_h$ holds by the definition.
  Hence, we have 
  \begin{align*}
  f_h(\Ch_h(X')) = \max\{ f_h&(X''')\mid X'''\subseteq X',~w_h(X''')\le w_h(\Ch_h(X'))\}.
  \end{align*}
  Thus, $\Ch_h$ is COM. 
\end{proof}

\begin{lemma}\label{lem:propbound}
  For each choice function defined in Algorithm \ref{alg:proputil} and 
  a set of contracts $X'\subseteq X_h$, it holds that 
  \begin{align*}
    w_h(\Ch_h(X'))< B_h+\upperW_h. 
  \end{align*}
\end{lemma}
\begin{proof}
  As in the proof for Lemma~\ref{lem:proputil}, 
  let $X'$ be $\{x^{(1)},\dots,x^{(|X'|)}\}\subseteq X_h$ such that 
  $x_{W}^{(1)}<\dots<x_{W}^{(|X'|)}$
  and let $k$ be the largest integer that satisfies $\sum_{i=1}^{k}x_W^{(i)}\le B_h$ and $k\le |X'|-1$.
  Then, we have 
  \begin{align*}
    w_h(\Ch_h(X'))
     = w_h(\{x^{(1)},\dots,x^{(k)}\})+w_h(\{x^{(|X'|)}\})
     < B_h+\upperW_h. 
  \end{align*}
  Thus, $w_h(\Ch_h(X'))< B_h+\upperW_h$ holds for any $h\in H$.
\end{proof}

From these two lemmas, we get the following theorem.
\begin{theorem}
  The generalized DA mechanism with the choice functions defined in Algorithm~\ref{alg:proputil} is
  strategy-proof for doctors and
  produces a $B_H'$-stable matching such that $B_h\le B_h'< B_h+\upperW_h$ for any $h\in H$
  when each hospital has utility over a set of contracts that is proportional to the total amount of wages. 
\end{theorem}

We further examine what happens if we give up strategy-proofness.
We can construct a stable mechanism that may increase the budget by a factor of up to one-half. 
This bound improves the previous bound when the maximum wage $\upperW_h$ is larger than $B_h/2$.
We here modify Algorithm~\ref{alg:nspchoice} to 
i) sort $X'$ in increasing order of wage; 
ii) pick the contract with the highest wage; and 
iii) pick the remaining contracts in order while keeping the total wage less than $1.5B_h$. 
Formally, we define this modification as Algorithm \ref{alg:proputil2}. 

\begin{algorithm}
\SetKwInOut{Input}{input}\Input{$X'\subseteq X_h$\quad\textbf{output:} $\Ch_h(X')$}
\caption{}\label{alg:proputil2}
Initialize $Y\ot \emptyset$\;
Sort $X'$ in increasing order of wages\; %
add the $|X'|$th contract (highest wage contract) to $Y$\;
\For{$i=1,2,\dots,|X'|-1$}{
  let $x$ be the $i$th contract in $X'$\;
  \lIf{$w_h(Y\cup\{x\}) < 1.5\cdot B_h$}{$Y\ot Y\cup\{x\}$}
}
\Return $Y$\;
\end{algorithm}

\begin{lemma}
  The generalized DA with the choice functions defined in Algorithm \ref{alg:proputil2} 
  is SUB, IRC, and COM. 
\end{lemma}

\begin{proof}
  IRC is evident 
  by the definition of the choice function. 

  Let $X'=\{x^{(1)},\dots,x^{(|X'|)}\}\subseteq X_h$ such that
  \(x^{(1)}_W<\dots<x^{(|X'|)}_W\)
  and let $k$ be the largest integer that satisfies
  \(x_W^{|X'|}+\sum_{i=1}^{k}x_W^{(i)}< 1.5\cdot B_h\) and $k\le |X'|-1$.
  Also, let 
  \(X''=\{x^{(j_1)},\dots,x^{(j_s)}\}\subseteq X'\)
  such that $1\le j_1<\dots<j_s\le |X'|$ 
  and let $k'$ be the largest index that satisfies
  \(x_W^{j_s}+\sum_{i=1}^{k'}x_W^{(j_i)}< 1.5\cdot B_h\) and \(k'\le s-1\).
  We then have $k\le j_{k'}$ by $x_W^{(i)}\le x_W^{(j_{i})}$.  
  $\Ch_h$ is SUB because
  \begin{align*}
    \Ch_h(X')\cap X''
    &=\{x^{(1)},\dots,x^{(k)},x^{(|X'|)}\}
    \cap\{x^{(j_{1})},\dots,x^{(j_{s})}\}\\
    &\subseteq \{x^{(j_1)},\dots,x^{(j_{k'})},x^{(j_s)}\}=\Ch_h(X'').
  \end{align*}

  We next claim that $\Ch_h$ is COM.
  Let $t$ satisfy that
  \begin{align*}
    x^{(1)}_W<\dots<x^{(t)}_W \le 1/2 < x^{(t+1)}_W<\dots<x^{(|X'|)}_W.
  \end{align*}
  We consider the following three cases. 
  
  \noindent\textbf{Case 1:} $|\Ch_h(X')\cap \{x^{(t+1)},\dots,x^{(|X'|)}\}|\ge 2$. 
  In this case, since $w_h(\Ch_h(X'))>2\cdot B_h/2=B_h$, 
  we have 
  \begin{align*}
    f_h(\Ch_h(X'))  = \max\{&f_h(X''')\mid X'''\subseteq X',\ w_h(X''') \le w_h(\Ch_h(X'))\}.  
  \end{align*}

  \noindent\textbf{Case 2:} $k<t$. 
  In this case, since $w_h(\Ch_h(X'))+w_h(\{x^{(t)}\})>1.5\cdot B_h$, 
  we have \[w_h(\Ch_h(X'))>1.5\cdot B_h-w_h(\{x^{(t)}\})\ge B_h.\]
  Thus, it holds that
  \begin{align*}
    f_h(\Ch_h(X'))  = \max\{&f_h(X''')\mid X'''\subseteq X',\ w_h(X''') \le w_h(\Ch_h(X'))\}.
  \end{align*}
  
  \noindent\textbf{Case 3:} $k\ge t$ and $|\Ch_h(X')\cap \{x^{(t+1)},\dots,x^{(|X'|)}\}|\le 1$.
  As $x^{(|X'|)}\in \Ch_h(X')$, we have $k=t$ and $|\Ch_h(X')\cap \{x^{(t+1)},\dots,x^{(|X'|)}\}|=1$.
  If $w_h(\Ch_h(X'))\ge B_h$, we have
  \begin{align*}
    f_h(\Ch_h(X'))=\max\{&f_h(X''')\mid X'''\subseteq X',\ w_h(X''') \le w_h(\Ch_h(X'))\}.  
  \end{align*}
  Thus, we assume that $w_h(\Ch_h(X'))< B_h$.
  In this case, we obtain \[f_h(\Ch_h(X'))=\max\{f_h(X''')\mid X'''\subseteq X',~w_h(X''')\le B_h\}\]
  because $|X'''\cap \{x^{(t+1)},\dots,x^{(|X'|)}\}|\le 1$ holds for any $X'''\subseteq X'$ such that $w_h(X''')\le B_h$.

  
  The proof is thus complete. 
\end{proof}

\begin{lemma}
  For each choice function defined in Algorithm \ref{alg:proputil2} and
  a set of contracts $X'\subseteq X_h$, it holds that 
  \begin{align*}
    w_h(\Ch_h(X'))\le 1.5\cdot B_h.
  \end{align*}
\end{lemma}
This lemma is obvious by the definition of the choice functions. 
Thus, $B_h\le B_h'\le 1.5\cdot B_h$ holds for any $h\in H$.

We summarize the above arguments in the following theorem: 
\begin{theorem}
  The generalized DA mechanism with the choice functions defined in Algorithm~\ref{alg:proputil2}
  produces a $B_H'$-stable matching such that
  $B_h\le B_h'\le 1.5\cdot B_h$ for any $h\in H$ 
  when each hospital has utility over a set of contracts that is proportional to the total amount of wages. 
\end{theorem}

Second, we consider the case in which each hospital has the same utility across contracts. 
Formally, for every $h$, $X'\subseteq X_h$ and a constant $\gamma_h(>0)$, 
\[f_h(X')=\gamma_h\cdot |X'|\] holds. 
In this case, we obtain a strategy-proof mechanism that always produces
a conventional stable matching, which never violates the given budget constraints. 
The choice function greedily chooses contracts in increasing order of wage
until just before the total wage of the chosen contracts exceeds the constraint. 

\begin{algorithm}
\SetKwInOut{Input}{input}\Input{$X'\subseteq X_h$\quad\textbf{output:} $\Ch_h(X')$}
\caption{}\label{alg:equtil}
Initialize $Y\ot \emptyset$\;
Sort $X'$ in increasing order of wages\;
\For{$i=1,2,\dots,|X'|$}{
  let $x$ be the $i$th contract in $X'$\;
  \lIf{$w_h(Y\cup\{x\}) \le B_h$}{$Y\ot Y\cup\{x\}$}
}
\Return $Y$\;
\end{algorithm}

\begin{theorem}
The generalized DA with the choice functions defined in Algorithm \ref{alg:equtil} is strategy-proof and 
produces a stable matching.
\end{theorem}

\begin{proof}
Since the choice functions greedily choose contracts in increasing order of wages,
we have
\[|\Ch_h(X')|=\max\{|X''|\mid X''\subseteq X',~w_h(X'')\le B_h\}\] 
for all $h$ and $X'\subseteq X_h$.
In other words,
\[\Ch_h(X')\in\argmax\{f_h(X'')\mid X''\subseteq X',~w_h(X'')\le B_h\}\] holds.
Thus, for each $h$ and $X'\subseteq X_h$,
$\Ch_h(X')$ picks the smallest \(\max\{|X''|\mid X''\subseteq X',~w_h(X')\le B_h\}\) contracts
according to the wage. 
Therefore, each $\Ch_h$ satisfies SUB, IRC, and LAD, 
and hence the mechanism is strategy-proof for doctors. 
Also, the mechanism clearly produces a $B_H$-feasible set of contracts. 
\end{proof}

\end{modified}

\section{Conclusion}
\label{sec:conclusion}
This paper examined how the matching with contract framework could be adjusted in the 
presence of budget constraints. Focusing on the idea of near-feasible matchings, 
our analysis established a useful characterization. 
To demonstrate this advantage, we proposed two novel mechanisms that 
return a near-feasible stable matching in polynomial time: 
one is strategy-proof and the other is not. 
Furthermore, we derived the increment bound of the budgets. 
Notably, the best possible bound is obtained by sacrificing strategy-proofness. 
 
%
While one might think that we could handle budget constraints together with maximum quota constraints,
it is not a simple process to design an appropriate choice function that handles both constraints simultaneously. 
Suppose, for example, that a hospital has a maximum quota of one and offers two contracts, 
one with a lower wage than the other. 
The hospital has lower utility for the low wage contract than the high wage one and, conversely, 
it has higher utility per unit wage for the former than the latter. 
Thus, the unique stable matching admits the high wage contract only. 
%
%
If the choice function additionally checks whether the current number of 
chosen contracts exceeds the maximum quota in line 5 in Algorithm~\ref{alg:nspchoice}, 
it chooses the low wage contract 
and fails to provide the stable matching. 
In general, this problem is known as a \textit{cardinality constrained knapsack problem}~\cite{capara2000,kellerer2004kp}.
The question remains whether building upon techniques for this problem 
we can construct a proper choice function. 
%
 
\section*{Acknowledgment}
This paper is based on the authors' conference publication~\cite{kawase:ijcai:2017}. It has been generalized to the assumption that each hospital's preference is represented by the weak relation over the subsets of contracts in Section~\ref{sec:com} and expanded to include two special cases in Section~\ref{sec:special cases}. Also, we added an impossibility example in Section~\ref{sec:impossibility} and placed the detailed proof of the intractability result in Section~\ref{sec:appendix}, and discuss doctor-side optimality in Section~\ref{sec:doctor-optimality}. 
This work was partially supported by JSPS KAKENHI Grant Number 26280081, 16K16005, and 17H01787.
%

\bibliographystyle{abbrv}
\bibliography{approxmatching} 

\appendix
\section{Proof of Theorem~\ref{thm:lowerapprox}}
\label{sec:appendix}

\begin{proof}
\begin{modified}
  Let $m$ be a positive integer larger than $1/(\beta-\alpha)+1/(1-\beta)$.
  We consider a market with $m^2$ doctors
  \begin{align*}
    D=\{d^*,d_1^0,\dots,d_1^m,d_2^0,\dots,d_2^m,\dots,d_{m-1}^0,\dots,d_{m-1}^m\}
  \end{align*}
  and $m$ hospitals $H=\{h_1,h_2,\dots,h_m\}$. 
  Figure~\ref{fig:lowerapprox} illustrates the market with some doctors' preferences and hospitals' utilities. 
  The set of offered contracts is a union of 
  \begin{align*}
    X_{d^*}&=\{(d^*,h_m,\beta)\},\\
    X_{d_i^0}&=\{(d_i^0,h_m,1/m), (d_i^0,h_i,\beta)\}\ \text{for all}\ i\in [m-1],\\
    X_{d_i^j}&=\{(d_i^j,h_i,\tfrac{1-\beta}{m-1})\}\ \text{for all}\ i,j\in [m-1],\ \text{and} \\
    X_{d_i^m}&=\{(d_i^m,h_i,\beta),(d_i^m,h_m,1/m)\}\ \text{for all}\ i\in [m-1], 
  \end{align*} 
  where $[m-1]$ indicates $\{1,2,\dots,m-1\}$.
  We assume that the preferences of the doctors are 
  \begin{align*}
    \succ_{d^*}:&~(d^*,h_m,\beta),\\
    \succ_{d_i^0}:&~(d_i^0,h_m,1/m)\succ_{d_i^0}   (d_i^0,h_i,\beta)\quad \text{for all}\ i\in [m-1],\\
    \succ_{d_i^j}:&~(d_i^j,h_i,\tfrac{1-\beta}{m-1})\ \text{for all}\quad i,j\in [m-1],\ \text{and}\ \\
    \succ_{d_i^m}:&~(d_i^m,h_i,\beta)\succ_{d_i^m} (d_i^m,h_m,1/m)\ \text{for\ all}\ i\in [m-1].
  \end{align*}
  The utilities of each hospital are 
  \begin{align*}
    f_{h_m}((d^*,h_m,\beta))&=1,\\
    f_{h_m}((d_i^0,h_m,1/m))&=2^{-i}\quad \text{for all}\ i\in [m-1],\\
    f_{h_i}((d_i^0,h_i,\beta))&=2^m\quad \text{for all}\ i\in [m-1],\\
    f_{h_i}((d_i^j,h_i,\tfrac{1-\beta}{m-1}))&=2^{m-j}\quad \text{for\ all}\ i,j\in [m-1],\\
    f_{h_i}((d_i^m,h_i,\beta))&=1\quad \text{for all}\ i\in [m-1],\ \text{and}\\
    f_{h_m}((d_i^m,h_m,1/m))&=2^{m-i}\quad \text{for all}\ i\in [m-1], 
  \end{align*}
  assuming they are additive. 
  Every hospital has the fixed budget $1$~($B_{h_1}=\dots=B_{h_m}=1$). 
  The condition $\upperW_h\le\beta\cdot B_h$ holds since $1/m\le\beta=\beta\cdot B_h$ and $(1-\beta)/(m-1)\le\beta=\beta\cdot B_h$.
  The market $(D,H,X,\succ_D,f_H,B_H)$ is depicted in Figure \ref{fig:lowerapprox}.

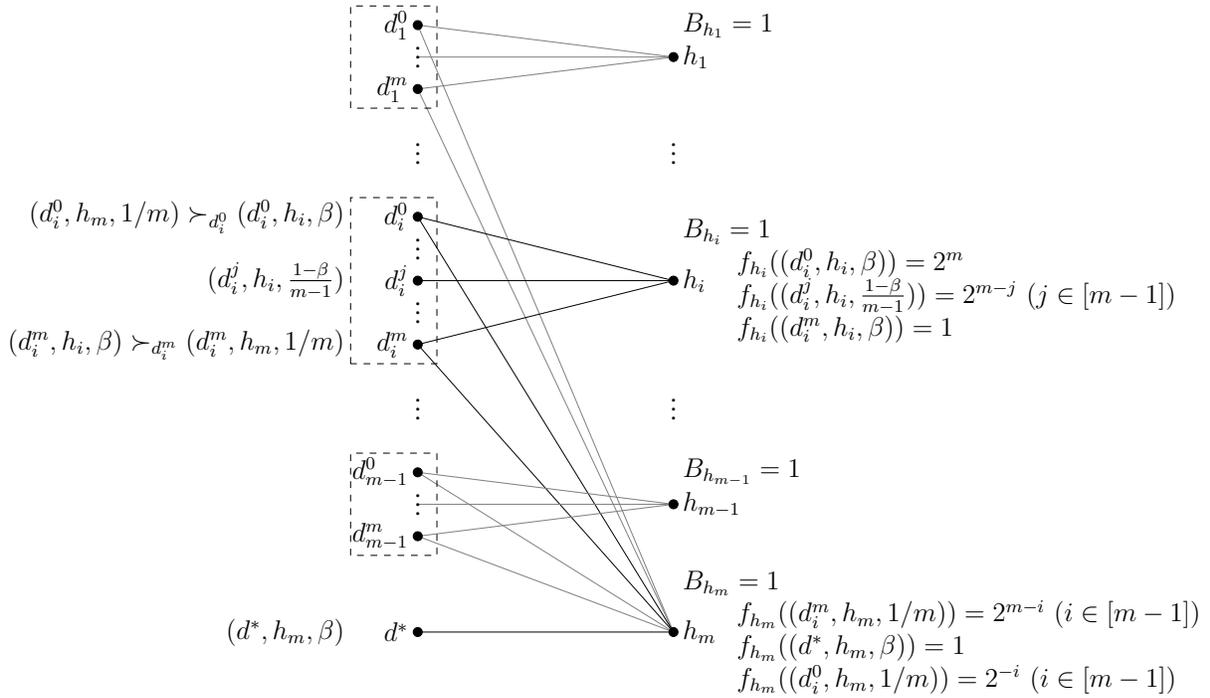
\begin{figure}[t]
\begin{center}
\scalebox{0.85}{
\begin{tikzpicture}[main node/.style={circle,draw,inner sep=1pt}]
\draw[ultra thin,gray] (2,9) -- (6,-0.5);
\draw[ultra thin,gray] (2,8) -- (6,-0.5);
\draw[] (2,6) -- (6,-0.5);
\draw[] (2,4) -- (6,-0.5);
\draw[ultra thin,gray] (2,2) -- (6,-0.5);
\draw[ultra thin,gray] (2,1) -- (6,-0.5);
\draw[] (2,-0.5) -- (6,-0.5);

\draw[ultra thin,gray] (2,8) -- (6,8.5);
\draw[ultra thin,gray] (2,8.5) -- (6,8.5);
\draw[ultra thin,gray] (2,9) -- (6,8.5);

\draw[] (2,4) -- (6,5);
\draw[] (2,5) -- (6,5);
\draw[] (2,6) -- (6,5);

\draw[ultra thin,gray] (2,1) -- (6,1.5);
\draw[ultra thin,gray] (2,1.5) -- (6,1.5);
\draw[ultra thin,gray] (2,2) -- (6,1.5);

\draw [dashed] (.95,7.7) rectangle (2.3,9.3);
\filldraw[] (2,9) circle (2pt) node[left,black] {$d_{1}^0$};
\node at (2,8.6) {$\vdots$};
\filldraw[] (2,8) circle (2pt) node[left,black] {$d_{1}^m$};
  
\node at (2,7.1) {$\vdots$};
\draw [dashed] (.95,3.7) rectangle (2.3,6.3);
\filldraw[] (2,6) circle (2pt) node[left,black] {$d_{i}^0$};
\node[left] at (1,6) {$(d_i^0,h_m,1/m)\succ_{d_i^0}(d_i^0,h_i,\beta)$};
\node at (2,5.6) {$\vdots$};
\filldraw[] (2,5) circle (2pt) node[left,black] {$d_{i}^j$};
\node[left] at (1,5) {$(d_i^j,h_i,\frac{1-\beta}{m-1})$};
\node at (2,4.6) {$\vdots$};
\filldraw[] (2,4) circle (2pt) node[left,black] {$d_{i}^m$};
\node[left] at (1,4) {$(d_i^m,h_i,\beta)\succ_{d_i^m}(d_i^m,h_m,1/m)$};

\node at (2,3.1) {$\vdots$};
\draw [dashed] (.95,0.7) rectangle (2.3,2.3);
\filldraw[] (2,2) circle (2pt) node[left,black] {$d_{m-1}^0$};
\node at (2,1.6) {$\vdots$};
\filldraw[] (2,1) circle (2pt) node[left,black] {$d_{m-1}^m$};

\filldraw[] (2,-0.5) circle (2pt) node[left,black] {$d^*$};
\node[left] at (1,-0.5) {$(d^*,h_m,\beta)$};

\filldraw[] (6,8.5) circle (2pt) node[right,black] {$h_1$};
\node[right] at (6,9) {$B_{h_{1}}=1$};
\node at (6,7.1) {$\vdots$};
\filldraw[] (6,5) circle (2pt) node[right,black] {$h_i$};
\node[right, text width=9cm] at (6,5) {$B_{h_{i}}=1$\\ \qquad$f_{h_i}((d_i^0,h_i,\beta))=2^m$\\ \qquad$f_{h_i}((d_i^j,h_i,\frac{1-\beta}{m-1}))=2^{m-j}$ $(j\in [m-1])$\\ \qquad$f_{h_i}((d_i^m,h_i,\beta))=1$};
\node at (6,3.1) {$\vdots$};
\filldraw[] (6,1.5) circle (2pt) node[right,black] {$h_{m-1}$};
\node[right] at (6,2) {$B_{h_{m-1}}=1$};

\filldraw[] (6,-0.5) circle (2pt) node[right,black] {$h_m$};
\node[right, text width=9cm] at (6,-0.5) {$B_{h_m}=1$\\ \qquad$f_{h_m}((d_i^m,h_m,1/m))=2^{m-i}$ $(i\in [m-1])$\\ \qquad$f_{h_m}((d^*,h_m,\beta))=1$\\ \qquad$f_{h_m}((d_i^0,h_m,1/m))=2^{-i}$ $(i\in [m-1])$};

\end{tikzpicture}}
\end{center}
\caption{Example of a market with no $B_H'$-stable matching such that $B_h\le B_h'\le (1+\alpha)\cdot B_h$ for all $h\in H$.} \label{fig:lowerapprox}
\end{figure}
  
We are going to show that there is no $B_H'$-stable matching 
such that $B_h\le B_h'\le (1+\alpha)\cdot B_h$ for all $h\in H$ by contradiction. 
Let us assume that $X'$ is a $B_H'$-stable matching, namely,
a stable matching for a market $(D,H,X,\succ_D,f_H,B_H')$ such that $B_h\le B_h'\le (1+\alpha)B_h$ for all $h\in H$. 
    
  First, let us consider a case where, for all $i\in [m-1]$, 
  doctor $d_i^0$ is assigned to $h_m$ with $1/m$, i.e., 
  \[
  Y_m^0\equiv\{(d_1^0,h_m,1/m), \ldots, (d_{m-1}^0,h_m,1/m) \}\subseteq X'. 
  \]
  In this case, $d^*$ must be assigned to $h_m$ because we assume $X'$ is $B'_H$-stable. 
  Specifically, doctor $d^*$ prefers $h_m$ to being unmatched. 
  Hospital $h_m$ prefers the contract with $d^*$, whose utility is one, to $Y_m^0$, 
  whose utility is $\sum_{i=1}^{m-1}2^{-i}$. 
  In addition, 
  the wage to $d^*$, that is, $\beta$, is smaller than the total wages to $Y_m^0$, that is, $(m-1)/m$. 
  Thus, unless $d^*$ is assigned to $h_m$ with $\beta$, she can form a blocking coalition. 
  Since $X'$ contains at least $(d^*,h_m,\beta)$ and $Y_m^0$, we derive 
  \begin{align*}
    w_{h_m}(X') & \ge \beta+(m-1)/m\\ 
                & > 1+\beta -\frac{(\beta-\alpha)(1-\beta)}{1-\alpha} \\
                & > (1+\beta)-(\beta-\alpha) = (1+\alpha) 
  \end{align*}
  from the assumptions of $m$, $\alpha$, and $\beta$. 
  Thus, $w_{h_m}(X')$ is strictly greater than $(1+\alpha)B_{h_{m}}$, 
  contradicting that $X'$ is $B_H'$-feasible, which is implied by $B'_H$-stability. 

  Second, let us consider another case where, for some $i\in [m-1]$, 
  $d_i^0$ is not assigned to $h_m$, i.e., $Y_m^0\not\subseteq X'$. 
  In this case, $d_i^0$ must be assigned to $h_i$. 
  To illustrate this, it is sufficient to consider a situation where $(d_i^m,h_i,\beta)$ and 
  \[
  Y_i \equiv \{(d_i^1,h_i,\tfrac{1-\beta}{m-1}), \ldots, (d_{i}^{m-1},h_i,\tfrac{1-\beta}{m-1}) \}
  \] 
  are chosen by $h_i$, i.e., $\{(d_i^m,h_i,\beta)\}\cup Y_i\subseteq X'$. 
  Hospital $h_i$ obtains the utility of $1 + \sum_{i=1}^{m-1}2^{m-i}~(=2^m-1)$ on the assignment. 
  On the other hand, if $d_i^0$ is assigned to $h_i$ with $\beta$, $h_i$ obtains the utility of $2^m$. 
  Evidently, $h_i$ prefers $(d_i^0,h_i,\beta)$ to $\{(d_i^m,h_i,\beta)\}\cup Y_i$ 
  and $d_i^0$ can form a blocking coalition unless she is assigned to $h_i$. 
  %
  To consider a set of doctors who are not assigned to $h_m$, 
  we introduce a set of indexes
  \[
  I=\{i\in[m-1]\mid (d_i^0,h_m,1/m)\not\in X'\}.
  \]
  By the assumption, $I$ is not the empty set. 
  In what follows, we concentrate on matchings where doctor $d_i^0$, for all $i\in I$, is assigned to $h_i$ instead of $h_m$. 
  Note that $d_i^m$ for $i\in [m-1]\setminus I$ must be assigned to $h_i$ 
  because, if $d_i^m$ is not assigned to $h_i$, 
  $\{(d_i^m,h_i,\beta)\}\cup Y_i$ is a blocking coalition. 
  In addition, $d_i^m$ for $i\in I$ must be assigned to $h_i$ or $h_m$ 
  because, if $d_i^m$ is unmatched, 
  $X_h\setminus\{(d^*,h_m,\beta)\}\cup \{(d_i^m,h_m,1/m)\}$ is a blocking coalition.

  Let us first examine a case where, for all $i\in I$, $d_i^m$ are assigned to $h_m$. 
  In this case, $d^*$ must be assigned to $h_m$. 
  Unless $d^*$ is assigned to $h_m$, $X''=X_{h_m}'\cup\{(d_i^0,h_m,1/m)\}$ is a blocking coalition for $i\in I$. 
  This is because $d_i^0$ prefers $h_m$ the most and $w_{h_m}(X'')=1$. 
  Then, since $X'_{h_m}$ contains 
  \begin{align*}
    \{ (d_{i}^0,h_m,1/m) \mid i\in [m-1]\setminus I \},\  
    \{ (d_i^m,h_m,1/m) \mid i\in I \},\ \text{and}\ (d^*,h_m,\beta), 
  \end{align*}
  we derive $w_{h_m}(X')=(m-1)/m+\beta$, which is strictly greater than $1+\alpha$
  from the assumptions of $m$, $\alpha$, and $\beta$. 
  Thus, $w_{h_m}(X')$ is strictly greater than $(1+\alpha)B_{h_{m}}$, which contradicts that $X'$ is $B'_H$-stable. 

  Next, let us examine another case where, for some $i\in I$, doctor $d_i^m$ is assigned to $h_i$. 
  In this case, $Y_i$ 
  must be chosen by $h_i$ because, if $d_i^j$ is not assigned to $h_i$ for some $j\in[m-1]$, 
  $X_{h_i}'\setminus \{(d_i^m,h_i,\beta)\}\cup\{(d_i^j,h_i,\frac{1-\beta}{m-1})\}$ is a blocking coalition. 
  In fact, doctor $d_i^j$ prefers $h_i$ to being unmatched. 
  Hospital $h_i$ prefers $(d_i^j,h_i,\frac{1-\beta}{m-1})$, whose utility is $2^{m-j}~(\ge 2)$, to $(d^m_i,h_i,\beta)$, whose utility is $1$. 
  In addition, the wage to $d_i^j$, that is, $\frac{1-\beta}{m-1}$, is smaller than the wage to $d^m_i$, that is, $\beta$. 
  Then, since $X'_{h_i}$ contains $(d_i^0,h_i,\beta)$, $(d_i^m,h_i,\beta)$, and $Y_i$, 
  we derive $w_{h_i}(X') = 1 + \beta > 1 + \alpha$ from the assumptions of $m$, $\alpha$, and $\beta$. 
  Thus, $w_{h_i}(X')$ is strictly greater than $(1+\alpha)B_{h_{i}}$, contradicting that $X'$ is $B'_H$-stable. 
\end{modified}
\end{proof}

\end{document}